\documentclass[envcountsame,envcountsect,10pt,oribibl]{llncs}

\usepackage{microtype}

\usepackage[usenames,dvipsnames]{xcolor}

\usepackage{amssymb,amstext,amsfonts,amsopn,stmaryrd}

\usepackage{tabularx}
\usepackage{array,multirow} 

\usepackage{comment}

\usepackage{graphics}
\usepackage{amsmath}
\usepackage{subfig}			 
\usepackage{stackrel}
\usepackage{tikz}
\usepackage{mathrsfs}
\usepackage{enumerate}
\usepackage{amssymb}

\graphicspath{{./}}

\newcommand{\pname}[1]{\textnormal{\textsc{#1}}}
\newcommand{\cclass}[1]{\textnormal{\textsf{#1}}}

\newcommand{\HED}{\pname{$H$-free Edge Deletion}}
\newcommand{\HDED}{\pname{$H'$-free Edge Deletion}}

\newcommand{\HEC}{\pname{$H$-free Edge Completion}}
\newcommand{\TED}{\pname{$T$-free Edge Deletion}}
\newcommand{\TDED}{\pname{$T'$-free Edge Deletion}}
\newcommand{\PTED}{\pname{$P_3$-free Edge Deletion}}
\newcommand{\PFED}{\pname{$P_4$-free Edge Deletion}}
\newcommand{\SLED}{\pname{$S_{\ell}$-free Edge Deletion}}
\newcommand{\SLOED}{\pname{$S_{\ell-1}$-free Edge Deletion}}
\newcommand{\SLLED}{\pname{$S_{\ell_1,\ell_2}$-free Edge Deletion}}
\newcommand{\SLOLOED}{\pname{$S_{\ell_1-1,\ell_2-1}$-free Edge Deletion}}
\newcommand{\RED}{\pname{$R$-free Edge Deletion}}

\newcommand{\KTED}{\pname{$K_3$-free Edge Deletion}}
\newcommand{\SHED}{\pname{$sH$-free Edge Deletion}}
\newcommand{\THED}{\pname{$tH$-free Edge Deletion}}
\newcommand{\TOHED}{\pname{$(t-1)H$-free Edge Deletion}}
\newcommand{\HOED}{\pname{$H_1$-free Edge Deletion}}
\newcommand{\TKTED}{\pname{$tK_2$-free Edge Deletion}}
\newcommand{\TWKTED}{\pname{$2K_2$-free Edge Deletion}}
\newcommand{\HBEC}{\pname{$\overline{H}$-free Edge Completion}}
\newcommand{\HBED}{\pname{$\overline{H}$-free Edge Deletion}}
\newcommand{\REC}{\pname{$R$-free Edge Completion}}
\newcommand{\PIED}{\pname{$\Pi$ Edge Deletion}}
\newcommand{\PIEC}{\pname{$\Pi$ Edge Completion}}
\newcommand{\PIEE}{\pname{$\Pi$ Edge Editing}}
\newcommand{\PIVD}{\pname{$\Pi$ Vertex Deletion}}
\newcommand{\TSAT}{\pname{3-SAT}}
\newcommand{\VC}{\pname{Vertex Cover}}
\newcommand{\CLED}{\pname{$C_\ell$-free Edge Deletion}}

\newcommand{\NP}{\cclass{NP}}
\newcommand{\NPC}{\cclass{NP-complete}}
\newcommand{\NPH}{\cclass{NP-hard}}
\newcommand{\coNPpoly}{\cclass{coNP/poly}}

\DeclareMathOperator{\diam}{diam} 

\newtheorem{construction}{Construction}

\newtheorem{observation}[lemma]{Observation}

\frontmatter          
\begin{document}

\mainmatter              
\title{Parameterized lower bound and NP-completeness of some $H$-free Edge Deletion problems}
%
%
%
%
\author{N. R. Aravind\inst{1} \and R. B. Sandeep\inst{1}\thanks{supported by TCS Research Scholarship} \and Naveen Sivadasan\inst{2}}
\institute{Department of Computer Science \& Engineering\\
Indian Institute of Technology Hyderabad, India\\
\email{\{aravind,cs12p0001\}\makeatletter@\makeatother iith.ac.in}
\and
TCS Innovation Labs, Hyderabad, India\\
\email{naveen\makeatletter@\makeatother atc.tcs.com}}

\maketitle              
\begin{abstract}
For a graph $H$, the \HED\ problem asks whether there exist at most $k$ edges whose deletion from the input graph $G$ 
results in a graph without any induced copy of $H$. We prove that \HED\ is \cclass{NP-complete} if $H$ is a graph with at least two edges and
$H$ has a component with maximum number of vertices which is a tree or a regular graph.
Furthermore, we obtain that these \NPC\ problems cannot be 
solved in parameterized subexponential time, i.e., in time $2^{o(k)}\cdot |G|^{O(1)}$, unless Exponential Time Hypothesis fails.
\end{abstract}
\section{Introduction}
\label{sec:introduction}

Graph modification problems ask whether we can obtain a graph $G'$ from an input graph $G$ by
at most $k$ number of \emph{modifications} on $G$ such that $G'$ satisfies some properties. 
Modifications could be any kind of 
operations on vertices or edges.  
For a graph property $\Pi$, the \PIED\ problem is to check whether there exist at most $k$ edges 
whose deletion from the input graph results in a graph with property $\Pi$. \PIEC\ and \PIEE\
are defined similarly, where \textsc{Completion} allows only adding (completing) edges and \textsc{Editing}
allows both completion and deletion. Another graph modification problem is \PIVD,
where at most $k$ vertex deletions are allowed. 
The focus of this paper is on \HED. It asks whether 
there exist at most $k$ edges whose removal from the input graph $G$
results in a graph $G'$ without any induced copy of $H$. The corresponding
\textsc{Completion} problem \HEC\ is equivalent to \HBED\ where
$\overline{H}$ is the complement graph of $H$. Hence the results we 
obtain on \HED\ translate to that of \HBEC.

Graph modifications problems have been studied rigorously from 1970s onward. 
Initially, the studies
were focused on proving that a modification problem is \NPC\ or solvable in polynomial time.
These studies resulted a good yield for vertex deletion problems: Lewis and Yannakakis
proved \cite{lewis1980node} that \PIVD\ is \NPC\ if $\Pi$ is non-trivial and 
hereditary on induced subgraphs. In other words, \PIVD\ is \NPC\ if $\Pi$ is 
defined by a finite set of forbidden induced subgraphs. Interestingly, researchers 
could not find a dichotomy result for \PIED\ similar to that of \PIVD. 
The scarcity of hardness results for \PIED\ is mentioned in many
papers in the last four decades. For examples, see \cite{yannakakis1981edge} and \cite{drange2015trivially}.
It is a folklore result that \HED\ can be solved in polynomial time if $H$ is a graph with at most one edge. 
Only these \HED\ problems are known to have polynomial time algorithms.
Cai and Cai proved that \HED\ is incompressible if $H$ is 3-connected but not complete, 
and \HEC\ is incompressible if $H$ is 3-connected and has at least two non-edges,
unless \NP$\subseteq$\coNPpoly~\cite{cai2014incompressibilityj}.
Further, under the same assumption, 
it is proved that \HED\ and \HEC\ are incompressible if $H$ is a tree on at least 7 vertices, which is not a star graph and
\HED\ is incompressible if $H$ is the star graph $K_{1,s}$, where $s\geq 10$~\cite{cai2012polynomial}.
They use polynomial parameter transformations for the reductions. This implies that these problems
are \NPC. 
The \HED\ problems are \NPC\ where $H$ is $C_\ell$ for any fixed $\ell\geq 3$, claw ($K_{1,3}$) \cite{yannakakis1981edge},
$P_\ell$ for any fixed $\ell\geq 3$ \cite{el1988complexity}, $2K_2$ \cite{drange2014exploring} and
diamond ($K_4-e$) \cite{fellows2011graph}.
In this paper, we prove that \HED\ is \NPC\ if $H$ has at least two edges and has a component with maximum 
number of vertices
which is a tree or a regular graph.
For every such graph $H$, to obtain that \HED\ is \NPC, we compose a series of polynomial time reductions 
starting from the reductions from one of the four base problems: \PTED, \PFED, \KTED\ and \TWKTED.
We believe that this technique can be extended to obtain a dichotomy result - \HED\ is \NPC\
if and only if $H$ has at least two edges. The evidence for this belief is discussed in the concluding section.

Another active area of research is to give parameterized lower bounds for graph modification 
problems. For example, to prove that a problem cannot be solved in parameterized subexponential time,
i.e., in time $2^{o(k)}\cdot |G|^{O(1)}$, under some complexity theoretic assumption,
where the parameter $k$ is the size of the solution being sought. 
For this, the technique used is a linear parameterized reduction - a polynomial time reduction
where the parameter blow up is only linear - from a problem which is already known to have no
parameterized subexponential time algorithm under the Exponential Time Hypothesis (ETH). 
ETH is a widely believed complexity theoretic assumption that \TSAT\ cannot be solved in 
subexponential time, i.e., in time $2^{o(n)}$, where $n$ is the 
number of variables in the \TSAT\ instance. Sparsification
Lemma \cite{impagliazzo1998problems} implies that, under ETH, there exist no algorithm to solve 
\TSAT\ in time $2^{o(n+m)}\cdot (n+m)^{O(1)}$,
where $m$ is the number of clauses in the \TSAT\ instance. Sparsification Lemma
considerably helps to obtain linear parameterized reductions from \TSAT\ as it is allowed to have a 
parameter $k$ such that $k=O(m+n)$ in the reduced problem instance.
It is known that the base problems mentioned in the last paragraph 
cannot be solved in parameterized subexponential time, unless ETH fails.
Since all the reductions we 
introduce here are compositions of linear parameterized reductions from the base problems,
we obtain that \HED\ cannot be solved in parameterized subexponential time,
unless ETH fails, if 
$H$ is a graph with at least two edges and has a component with maximum number of vertices which is a tree or a regular graph. 

\begin{figure}[h]
\centering
\subfloat[Subfigure 1 list of figures text][$P_3$]{
\includegraphics[width=0.1\textwidth]{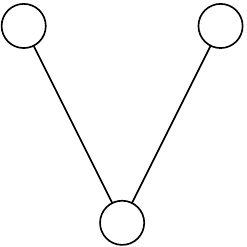}
\label{fig:subfig1}}
\qquad
\subfloat[Subfigure 2 list of figures text][$P_4$]{
\includegraphics[width=0.1\textwidth]{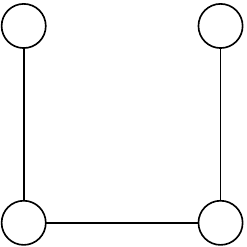}
\label{fig:subfig2}}
\qquad
\subfloat[Subfigure 2 list of figures text][$K_3$]{
\includegraphics[width=0.1\textwidth]{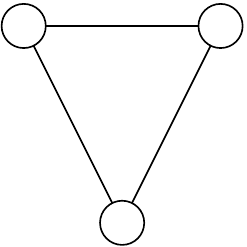}
\label{fig:subfig3}}
\qquad
\subfloat[Subfigure 2 list of figures text][$2K_2$]{
\includegraphics[width=0.1\textwidth]{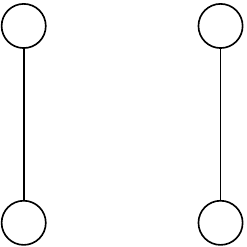}
\label{fig:subfig4}}
\caption{The four base problems are \PTED, \PFED, \KTED\ and \TWKTED.}
\label{fig:globfig}
\end{figure}

Graph modification problems have applications in DNA physical mapping 
\cite{bodlaender1996intervalizing,goldberg1995four},
numerical algebra \cite{rose1972graph}, circuit design \cite{el1988complexity}
and machine learning \cite{bansal2004correlation}.

\paragraph{Outline of the Paper:}
Section~\ref{sec:prebasics} gives the notations and terminology used in the paper.
It also introduces two constructions which are used for the reductions.
Section~\ref{sec:tree} proves that for any tree $T$ with at least two edges, 
\TED\ is \NPC\ and cannot be solved in
parameterized subexponential time, unless ETH fails.
Section~\ref{sec:regular} proves that for any connected regular graph $R$
with at least two edges, \RED\ is \NPC\ and cannot be solved in parameterized
subexponential time, unless ETH fails.
Section~\ref{sec:disconnected} combines the results from Sections~\ref{sec:tree}
and \ref{sec:regular} to prove that for any graph $H$ with at least two edges
such that $H$ has a component with maximum number of vertices which is a tree or a regular graph, \HED\ is \NPC\
and cannot be solved in parameterized subexponential time, unless ETH fails. 
As a consequence of the equivalence between \HED\ and \HBEC, we obtain 
the same results for \HBEC.
\section{Preliminaries and Basic Tools}
\label{sec:prebasics}

\paragraph{Graphs}: We consider simple, finite and undirected graphs.
The vertex set and the edge set of a graph $G$ is 
denoted by $V(G)$ and $E(G)$ respectively. 
$G$ is represented by the tuple $(V(G),E(G))$.
A simple path on 
$\ell$ vertices is denoted by $P_{\ell}$. 
For a vertex set $V'\subseteq V(G)$, $G[V']$ denotes 
the graph induced by $V'$ in $G$. $G-V'$ denotes the 
graph obtained by deleting all the vertices in $V'$ and
the edges incident to them from $G$. For an edge set $E'\subseteq E(G)$,
$G-E'$ denotes the graph $(V(G),E(G)\setminus E')$.
The diameter of a graph G, denoted by $\diam(G)$, is the number of edges
in the longest induced path in $G$.
An $r$-regular graph is a graph in which every vertex has degree $r$.
A regular graph is an $r$-regular graph for some non-negative integer $r$.
A dominating set of a graph $G$ is a set of vertices $V'\subseteq V(G)$
such that every vertex in $G$ is either in $V'$ or adjacent to 
at least one vertex in $V'$.
For a graph $G$, the disjoint union of $t$ copies of $G$ is denoted by $tG$. 
A component of a graph $G$ is a maximal connected subgraph of $G$.
A largest component of a graph is a component with maximum number of vertices.
We denote $|V(G)|+|E(G)|$ by $|G|$.
We follow \cite{douglas1996west} for further notations and terminology.

\paragraph{Technique for Proving Parameterized Lower Bounds}:
Exponential Time Hypothesis (ETH) is the assumption that \TSAT\ 
cannot be solved in time $2^{o(n)}$, where 
$n$ is the number of variables in
the \TSAT\ instance. Sparsification Lemma~\cite{impagliazzo1998problems}
implies that there exists no algorithm for \TSAT\ running in time $2^{o(n+m)}\cdot (n+m)^{O(1)}$, unless
ETH fails, where $n$ and $m$ are the number of variables and the number of clauses respectively
of the \TSAT\ instance. A linear parameterized reduction is a polynomial time reduction
from a parameterized problem $A$ to a parameterized problem $A'$ such that for every
instance $(G,k)$ of $A$, the reduction gives an instance $(G',k')$ of $B$ such that 
$k'=O(k)$. 

\begin{proposition}[\cite{platypus}]
  \label{pro:lpr}
  If there is a linear parameterized reduction from a parameterized problem $A$
  to a parameterized problem $B$ and if $A$ does not admit a parameterized subexponential
  time algorithm, then $B$ does not admit a parameterized subexponential time algorithm.
\end{proposition}

We refer the book \cite{platypus} for an excellent exposition on this and other aspects of 
parameterized algorithms and complexity.

\begin{proposition}
  \label{pro:bases}
  The following problems are \NPC. Furthermore, 
  they cannot be solved in time $2^{o(k)}\cdot |G|^{O(1)}$, unless ETH fails.
  \begin{enumerate}[(i)]
  \item\label{base:pted}  \PTED~\cite{komusiewicz2012cluster}
  \item\label{base:pfed}  \PFED~\cite{drange2014exploring}
  \item\label{base:cycle} \CLED\, for any fixed 
    $\ell\geq 3$~\cite{yannakakis1981edge}\footnote{{Yannakakis gives a polynomial time reduction from \VC\ to \CLED, for any fixed 
    $\ell\geq 3$ \cite{yannakakis1981edge}. If $\ell \neq 3$, the reduction he gives is a linear parameterized reduction. 
    When $\ell=3$, the reduction is not a linear parameterized reduction
    as it gives an instance with a parameter $k'=O(|E(G)|+k)$, 
    where $(G,k)$ is the input \VC\ instance. But, it is straight-forward to verify that composing
    the standard \TSAT\ to \VC\ 
    reduction (which is a linear parameterized reduction and gives a graph with $O(n+m)$ edges)
    with this reduction gives a linear parameterized reduction from \TSAT\ to \textsc{$K_3(C_3)$-free Edge Deletion}.}}
  \item\label{base:twkted} \TWKTED~\cite{drange2014exploring}
  \end{enumerate}
\end{proposition}

For any fixed graph $H$, the \HED\ problem trivially belongs to \NP. Hence, we may 
state that an \HED\ problem is \NPC\ by proving that it is \NPH.
\subsection{Basic Tools}

We introduce two constructions which will be used for the polynomial time reductions
in the upcoming sections.

\begin{construction}
  \label{con:nonadj}
  Let $(G',k, H, V')$ be an input to the construction, where $G'$ and $H$ are graphs, $k$
  is a positive integer and $V'$ is a subset of vertices of $H$.
  Label the vertices of $H$ such that every vertex get a unique label. Let the labelling be $\ell_H$.
  For every subgraph (not necessarily induced) $C$ with a vertex set $V(C)$ 
  and an edge set $E(C)$ in $G'$ such that $C$ is isomorphic to $H[V']$,
  do the following:
  \begin{itemize}
    \item Give a labelling $\ell_C$ for the vertices in $C$ such that there is an isomorphism
      $f$ between $C$ and $H[V']$ which maps every vertex $v$ in $C$ to a vertex $v'$ in $H[V']$
      such that $\ell_C(v)=\ell_H(v')$, i.e., $f(v)=v'$ if and only if $\ell_C(v)=\ell_H(v')$.
    \item Introduce $k+1$ sets of vertices $V_1, V_2,\ldots, V_{k+1}$, each of size $|V(H)\setminus V'|$.
    \item For each set $V_i$, introduce an edge set $E_i$ of size $|E(H)\setminus E(H[V'])|$ among
      $V_i\cup V(C)$
       such that there is an isomorphism $h$ between $H$ and
      $(V(C)\cup V_i, E(C)\cup E_i)$ which preserves $f$, i.e.,
      for every vertex $v\in V(C)$, $h(v)=f(v)$.
  \end{itemize}
  This completes the construction. Let the constructed graph be $G$. 
\end{construction}

An example of the construction is shown in Figure~\ref{fig:cons}.
Let $C$ be a copy of $H[V']$ in $G'$. Then, $C$ is called a \emph{base} in $G'$.
Let $\{V_i\}$ be the $k+1$ sets of vertices introduced in the construction for the base $C$.
Then, each $V_i$ is called a \emph{branch} of $C$ and the vertices in $V_i$ are called 
the \emph{branch vertices} of $C$. $C$ is called the \emph{base} of $V_i$ for $1\leq i\leq k+1$.
The vertex set of $G'$ in $G$ is denoted by $V_{G'}$.

Since $H$ is a fixed graph, the construction runs in polynomial time. 
In the construction, for every base $C$ in $G'$, we introduce new vertices
and edges such that there exist $k+1$ copies of $H$ in $G$ and $C$ is the common intersection of every pair 
of them. This enforces that every solution of an instance $(G,k)$ of \HED\ is a solution of 
an instance $(G',k)$ of \HDED, where $H'$ is $H[V']$. This is proved in the following lemma.

\begin{lemma}
  \label{lem:con:nonadj-backward}
  Let $G$ be obtained by Construction~\ref{con:nonadj} on 
  the input $(G',k,H,V')$, where $G'$ and $H$ are graphs, $k$ is a positive integer and $V'\subseteq V(H)$.
  Then, if $(G,k)$ is a yes-instance of \HED, then $(G',k)$ is a yes-instance of \HDED, where $H'$ is $H[V']$.
\end{lemma}
\begin{proof}
  Let $F$ be a solution of size at most $k$ of $(G,k)$. For a contradiction, assume that $G'-F$
  has an induced $H'$ with a vertex set $U$. Hence there is a base $C$ in $G'$ isomorphic to
  $H'$ with the vertex set $V(C)=U$. Since there are $k+1$ copies of $H$ in $G$, where each pair
  of copies of $H$ has the intersection $C$, and $|F|\leq k$, deleting $F$ cannot kill all the copies of 
  $H$ associated with $C$. Therefore,
  since $U$ induces an $H'$ in $G'-F$, there exists a branch $V_i$ of $C$ such that $U\cup V_i$
  induces $H$ in $G-F$, which is a contradiction.\qed
\end{proof}

\begin{figure}[h]
\centering
\subfloat[Subfigure 1 list of figures text][$G'$]{
\includegraphics[width=0.65in]{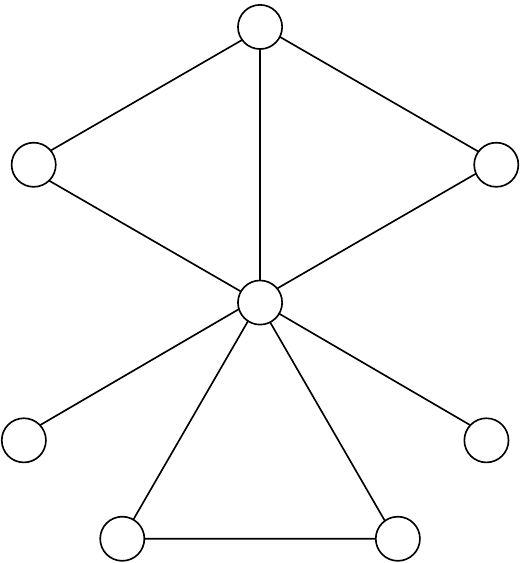}
\label{fig:con1g}}
\qquad
\subfloat[Subfigure 2 list of figures text][$H$. The vertices in $V'$ are blackened.]{
\includegraphics[width=0.65in]{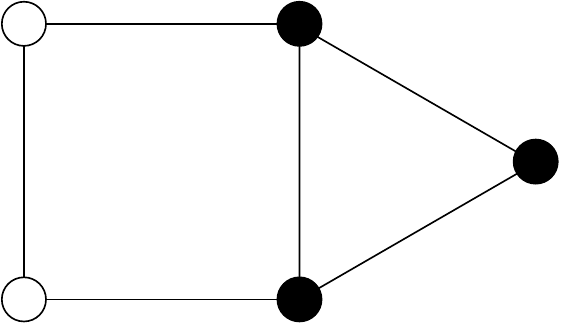}
\label{fig:con1h}}
\qquad
\subfloat[Subfigure 2 list of figures text][Output of Construction~\ref{con:nonadj} with an input $(G',k=2,H,V')$.]{
\includegraphics[width=1.15in]{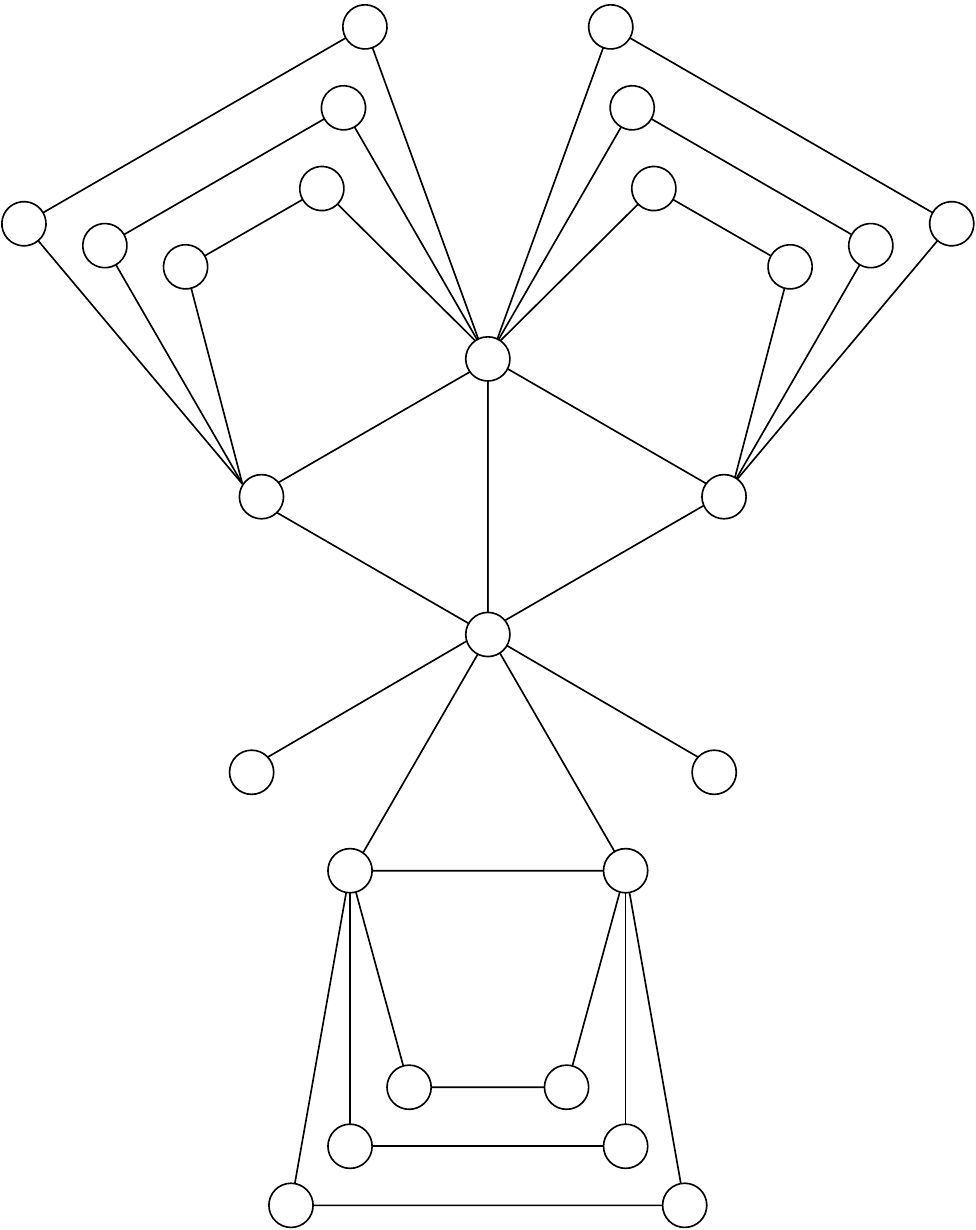}
\label{fig:constructed}}
\qquad
\subfloat[Subfigure 2 list of figures text][Output of Construction~\ref{con:adj} with an input $(G',k=2)$.]{
\includegraphics[width=1.15in]{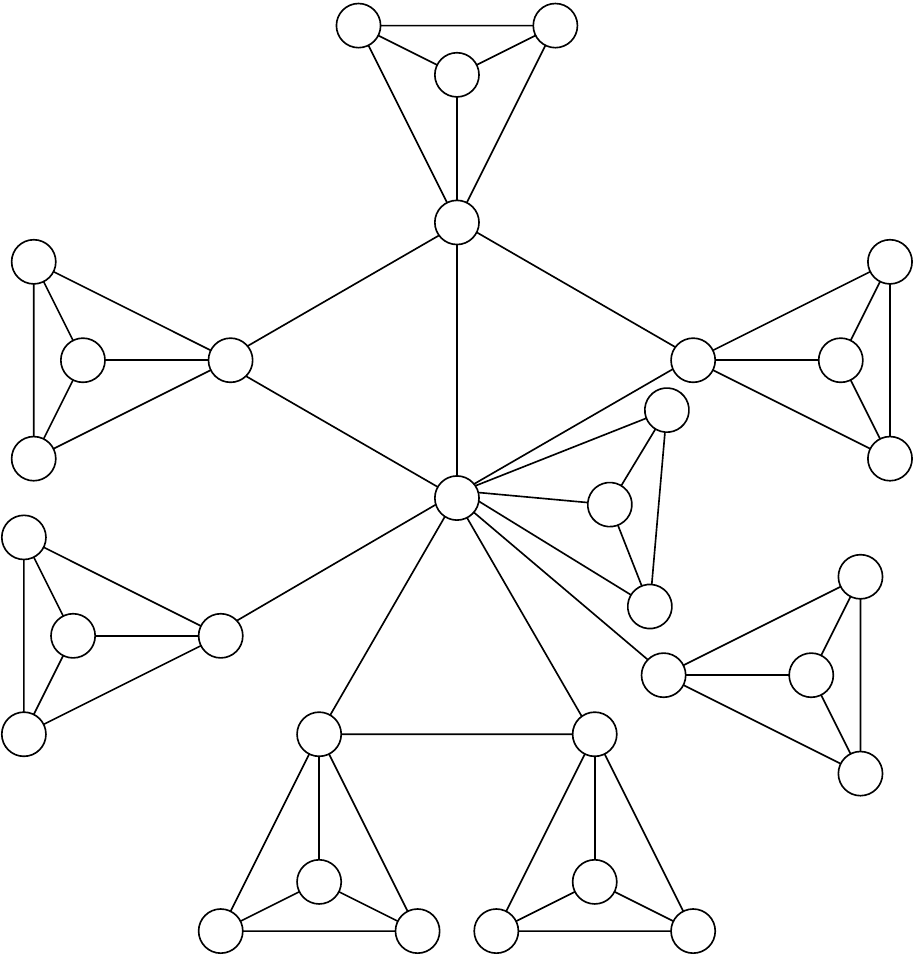}
\label{fig:constructed2}}
\caption{Examples showing Construction~\ref{con:nonadj} and Construction~\ref{con:adj}.}
\label{fig:cons}
\end{figure}

Now we introduce a simple construction, which is used in the next section. 
This construction attaches a clique of $k+1$ vertices to each vertex in the input
graph of the construction.

\begin{construction}
  \label{con:adj}
  Let $(G',k)$ be an input to the construction, where $G'$ is a graph and $k$
  is a positive integer.
  For every vertex $v_i$ in $G'$, introduce a set of $k+1$ vertices $V_i$ 
  and make every pair of vertices in $V_i\cup \{v_i\}$ adjacent.
  This completes the construction. Let the resultant graph be $G$.
\end{construction}

An example of the construction is shown in Figure~\ref{fig:cons}.
Here, we call all the newly introduced vertices as \emph{branch vertices}.
\section{\TED}
\label{sec:tree}

Let $T$ be any tree with at least two edges. We use induction on the diameter of $T$ to prove that \TED\ is \NPC. 
The base cases are when $\diam(T)=2$ or 3.  
To prove the base cases, we use polynomial time reductions from
\PTED\ and \PFED. For any $T$ with $\diam(T)>3$, we give polynomial time reduction from \TDED\ to \TED, where
$T'$ is a subtree of $T$ such that $\diam(T')=\diam(T)-2$. To prove each of the base cases,
we apply induction on the number of leaf vertices. All our reductions are linear parameterized reductions
and hence from the non-existence of parameterized subexponential algorithms for \PTED\ and \PFED\ , 
we obtain that there exists no parameterized subexponential time algorithm for \TED, unless ETH fails.

\subsection{Base Cases}
As mentioned above, the base cases are when $\diam(T)=2$ or 3. 
By $\ell(T)$, we denote the number of leaf vertices of $T$. 
We call the vertices in $T$ with degree one as \emph{leaf vertices} and the vertices with
degree more than one as \emph{internal vertices}.
If $\diam(T)=2$ and $\ell(T)=\ell\geq 2$, we denote $T$ by $S_\ell$, the star graph on $\ell+1$ vertices.

For every pair of non-negative integers $\ell_1$ and $\ell_2$ such that $\ell_1+\ell_2\geq 1$,
we define a tree denoted by $S_{\ell_1,\ell_2}$ as follows: the vertex set $V$ of $S_{\ell_1,\ell_2}$
has $\ell_1+\ell_2+2$ vertices with two designated adjacent vertices $r_1$ and $r_2$ such that 
$r_1$ is adjacent to $\ell_1$ number of leaf vertices in $V\setminus \{r_2\}$
and $r_2$ is adjacent to $\ell_2$ number of leaf vertices in $V\setminus \{r_1\}$. 
We call such a tree as a \emph{twin-star} graph.
We note that $S_{\ell_1,0}$ is the star graph $S_{\ell_1+1}$ and that $S_{\ell_1,\ell_2}$ and 
$S_{\ell_2,\ell_1}$ are isomorphic.

\begin{figure}[h]
\centering
\subfloat[Subfigure 1 list of figures text][$S_6$]{
\includegraphics[width=0.1\textwidth]{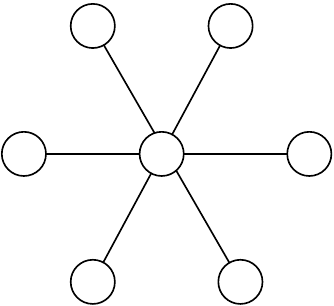}
\label{fig:subfig5}}
\qquad
\subfloat[Subfigure 2 list of figures text][$S_{5,2}$]{
\includegraphics[width=0.15\textwidth]{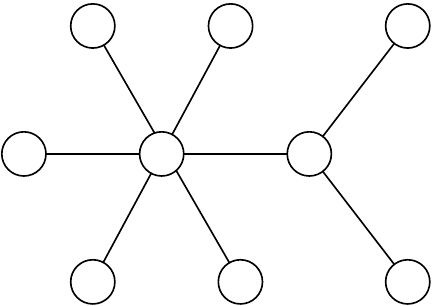}
\label{fig:subfig6}}
\caption{A star graph and a twin-star graph}
\label{fig:starfigs}
\end{figure}

\begin{lemma}
  \label{lem:sl-reduction}
  Let $\ell>2$. Then, there is a linear parameterized reduction from \SLOED\ 
  to \SLED.
\end{lemma}
\begin{proof}
  Let $(G',k)$ be an instance of \SLOED. 
  Apply Construction~\ref{con:adj} on $(G',k)$
  to obtain $G$. We claim that $(G',k)$ is a yes-instance of \SLOED\ if and only if
  $(G,k)$ is a yes-instance of \SLED.

  Let $(G',k)$ be a yes-instance of \SLOED. Let $F'$ be a solution of size at most $k$
  of $(G',k)$. For a contradiction, assume that $G-F'$ has an induced $S_{\ell}$
  with a vertex set $U$. Let $r$ be the internal vertex of the $S_\ell$ induced by $U$ in $G-F'$.
  Now there are two cases and in both the cases we obtain contradictions.
  \begin{itemize}
    \item $r$ is a branch vertex: 
      Since the neighborhood of any branch vertex in $G-F'$ is a clique, $r$ cannot be
      the internal vertex, which is a contradiction.
    \item $r$ is a vertex in $V_{G'}$:
      Since the branch vertices in the neighborhood of $r$ in $G-F'$
      induce a clique, at most one branch neighbor $u$ of $r$ is present in $U$ (as a leaf vertex).
      Hence, the remaining leaf vertices of the $S_\ell$ induced by $U$ in $G-F'$ belong to $V_{G'}$.
      This implies that $U\setminus \{u\}$ induces $S_{\ell-1}$ in $G'-F'$, which is a contradiction.
  \end{itemize}

  Conversely, let $(G,k)$ be a yes-instance of \SLED. Let $F$ be a solution of size at most $k$ of
  $(G,k)$. For a contradiction, assume that $G'-F$ has an induced $S_{\ell-1}$ with a vertex set $U$.
  Let $r$ be the internal vertex of $S_{\ell-1}$ induced by $U$ in $G'-F$. 
  Since $|F|\leq k$ and $k+1$ branch vertices are adjacent to $r$ in $G$, 
  there is at least one branch vertex $u$ adjacent to $r$ in $G-F$. 
  Hence, $U\cup \{u\}$ induces an $S_{\ell}$ in $G-F$, which is a contradiction.\qed
\end{proof}

\begin{theorem}
  \label{thm:star}
  For every integer $\ell\geq 2$, \SLED\ is \NPC. Furthermore, \SLED\ is not solvable in time
  $2^{o(k)}\cdot |G|^{O(1)}$, unless ETH fails.
\end{theorem}
\begin{proof}
  The proof is by induction on $\ell$. When $\ell=2$, $S_{\ell}$ is the graph $P_3$.
  Hence, Proposition~\ref{pro:bases}(\ref{base:pted}) proves this case.
  Assume that the statements are true for \SLOED, if $\ell-1\geq 2$. 
  Now the statements follow from Lemma~\ref{lem:sl-reduction}.\qed
\end{proof}

We apply a similar technique to prove the NP-completeness and parameterized lower bound for \TED\ 
when $\diam(T)=3$. As described before, we denote these graphs by $S_{\ell_1,\ell_2}$,
the twin-star graph having $\ell_1\geq 1$ leaf vertices adjacent to an internal vertex $r_1$ and $\ell_2\geq 1$ leaf
vertices adjacent to another internal vertex $r_2$. 

\begin{lemma}
  \label{lem:sll-reduction}
  For any pair of integers $\ell_1$ and $\ell_2$ such that $\ell_1,\ell_2\geq 1$ and $\ell_1+\ell_2\geq 3$, there
  is a linear parameterized reduction from \SLOLOED\ to \SLLED.
\end{lemma}
\begin{proof}
  Let $(G',k)$ be an instance of \SLOLOED. 
  Apply Construction~\ref{con:adj} on $(G',k)$ to obtain $G$.
  We claim that 
  $(G',k)$ is a yes-instance of \SLOLOED\ if and only if $(G,k)$ is a yes-instance of \SLLED.

  Let $(G',k)$ be a yes-instance of \SLOLOED. Let $F'$ be a solution of size at most $k$
  of $(G',k)$. For a contradiction, assume that $G-F'$ has an induced copy of $S_{\ell_1,\ell_2}$ with a 
  vertex set $U$. Let $r_1$ and $r_2$ be the two internal vertices of the $S_{\ell_1, \ell_2}$ induced by $U$ in $G-F'$.
  Now, there are the following cases and in each case, we obtain a contradiction.
  \begin{itemize}
  \item Either $r_1$ or $r_2$ is a branch vertex: This is not possible as the neighborhood
    of every branch vertex induces a clique in $G-F'$.
  \item Both $r_1$ and $r_2$ are in $V_{G'}$:
    Since the branch vertices adjacent to $r_1$ forms a clique in $G-F'$, at most one branch 
    vertex $u_1$ can be a leaf vertex adjacent to $r_1$ in the $S_{\ell_1,\ell_2}$ induced by $U$ in $G-F'$.
    Similarly, at most one branch vertex $u_2$ can be a leaf vertex adjacent to $r_2$
    in the $S_{\ell_1,\ell_2}$ induced by $U$ in $G-F'$.
    The remaining vertices of $U$ belong to $V_{G'}$. 
    Hence $U\setminus \{u_1,u_2\}$ induces $S_{\ell_1-1,\ell_2-1}$ in $G'-F'$, which is a contradiction.
  \end{itemize}

  Conversely, let $(G,k)$ be a yes-instance of \SLLED. Let $F$ be a solution of size at most $k$ of
  $(G,k)$. For a contradiction, assume that $G'-F$ has an induced $S_{\ell_1-1,\ell_2-1}$ with a vertex
  set $U$. Since $\ell_1+\ell_2\geq 3$, there exists at least one internal vertex, say $r_1$, 
  in the $S_{\ell_1-1,\ell_2-1}$
  induced by $U$ in $G'-F$. If there is no other internal vertex $r_2$ in the $S_{\ell_1-1,\ell_2-1}$,
  then let $r_2$ be any leaf vertex of the $S_{\ell_1-1,\ell_2-1}$.
  Let $V_1$ and $V_2$ be the set of branch vertices introduced in the construction such that 
  every vertex in $V_1$ is adjacent to $r_1$ and every vertex in $V_2$ is adjacent to $r_2$.
  Since $|F|\leq k$ and $|V_1|,|V_2|=k+1$, 
  there exist a vertex $v_1\in V_1$ adjacent to $r_1$ and a vertex $v_2\in V_2$
  adjacent to $r_2$ in $G-F$. Hence, $U\cup\{v_1,v_2\}$ induces an $S_{\ell_1, \ell_2}$ in $G-F$,
  which is a contradiction.\qed
\end{proof}

\begin{theorem}
  \label{thm:twin-star}
  For every pair of integers $\ell_1$ and $\ell_2$ such that $\ell_1,\ell_2\geq 0$ and $\ell_1+\ell_2\geq 1$, 
  \SLLED\ is \NPC\ and
  \SLLED\ is not solvable in time
  $2^{o(k)}\cdot |G|^{O(1)}$, unless ETH fails.
\end{theorem}
\begin{proof}
  The proof is by induction on $\ell_1+\ell_2$. The base cases are:
  \begin{itemize}
  \item $\ell_1=0$ ($\ell_2=0$): This is the case when the tree is $S_{\ell_2+1}$ ($S_{\ell_1+1}$),
    the case handled by Theorem~\ref{thm:star}.
  \item $\ell_1=\ell_2=1$: Here the tree is a $P_4$ and hence the statements follow 
    from Proposition~\ref{pro:bases}(\ref{base:pfed}).
  \end{itemize}  
  Assume that the statements holds true for the integers $\ell_1-1, \ell_2-1$ such 
  that $\ell_1-1,\ell_2-1\geq 0$ and $(\ell_1-1)+(\ell_2-1)\geq 1$.
  Now, the statements follow from Lemma~\ref{lem:sll-reduction}.\qed
\end{proof}

\subsection{Induction}
In the previous subsection, we proved the base cases of the inductive proof for the NP-completeness
and parameterized lower bound of \TED. The base cases were $\diam(T)=2$ (star graph) and $\diam(T)=3$ (twin-star graph).
Before concluding the proof, we give a lemma which is stronger than what we require and the further implications of this 
lemma will be discussed in the concluding section.

\begin{lemma}
  \label{lem:degree-reduction}
  Let $H$ be any graph and $d$ be any integer. Let $V'$ be the set of all
  vertices in $H$ with degree more than $d$. Let $H'$ be $H[V']$. Then, there is a
  linear parameterized reduction from
  \HDED\ to \HED.
\end{lemma}
\begin{proof}
  Let $(G',k)$ be an instance of \HDED. Obtain $G$ by applying 
  Construction~\ref{con:nonadj} on $(G',k,H,V')$. We claim that 
  $(G',k)$ is a yes-instance of \HDED\ if and only if 
  $(G,k)$ is a yes-instance of \HED.

  Let $(G',k)$ be a yes-instance of \HDED. Let $F'$ be a 
  solution of size at most $k$ of $(G',k)$. For a contradiction,
  assume that $G-F'$ has an induced $H$ with a vertex set $U$.
  Let $U'$ be the set of all vertices in $U$ such that every
  vertex in $U'$ has degree more than $d$ in $(G-F')[U]$.
  Since every branch vertex in $G$ has degree at most $d$,
  every vertex in $U'$
  must be in $V_{G'}$. Hence $U'$ induces an $H'$ in $G'-F'$,
  which is a contradiction.
  Lemma~\ref{lem:con:nonadj-backward} proves the converse.\qed
\end{proof}

Corollary~\ref{cor:tree-reduction} is obtained by invoking Lemma~\ref{lem:degree-reduction} with $H=T$ and $d=1$.

\begin{corollary}
  \label{cor:tree-reduction}
  Let $T$ be any tree with $\diam(T)>3$. 
  Let $T'$ be obtained from $T$ by deleting all leaf vertices. Then, there exists a 
  linear parameterized reduction from \TDED\ to \TED. 
\end{corollary}

\begin{theorem}
  \label{thm:tree}
  Let $T$ be any tree with at least two edges. Then, \TED\ is \NPC.
  Furthermore, \TED\ is not solvable in time
  $2^{o(k)}\cdot |G|^{O(1)}$, unless ETH fails.
\end{theorem}
\begin{proof}
  We apply induction on the diameter of $T$.
  Theorems~\ref{thm:star} and \ref{thm:twin-star} prove the statements 
  when $\diam(T)=2$ and $\diam(T)=3$ respectively.
  Let the statements be true when $\diam(T)=t'$ for all $t'$ 
  such that $2\leq t'\leq t$ for some $t\geq 3$. Assume that $T$ has diameter $t+1$.
  Deleting all leaf vertices from 
  $T$ gives a graph $T'$ with diameter $t+1-2=t-1\geq 2$. 
  Now the statements follow from Corollary~\ref{cor:tree-reduction}.\qed
\end{proof}
\section{\RED}
\label{sec:regular}
In this section, for any connected $r$-regular graph $R$, where $r>2$,
we give a direct reduction either from 
\PTED\ or from \KTED\ to \RED.
The following three observations are used to prove the reduction which is given in 
Lemma~\ref{lem:reg-reduction}.

\begin{observation}
  \label{obs:reg-dom}
  Let $R$ be an $r$-regular graph for some $r>2$. 
  Let $V'\subseteq V(R)$ be such that $|V'|=3$.
  Then, $V\setminus V'$ is a dominating set in $R$.
\end{observation}
\begin{proof}
  To prove that $V\setminus V'$ is a dominating set of $R$, we need to prove 
  that for every vertex $v\in V(R)$, either $v$ is in $V\setminus V'$ or 
  $v$ is adjacent to a vertex in $V\setminus V'$. If $v\notin V\setminus V'$,
  then $v\in V'$. Since $|V'| = 3$ and
  $v$ has degree $r\geq 3$, $v$ must have at least one edge to a vertex in $V\setminus V'$.\qed
\end{proof}

\begin{observation}
  \label{obs:reg-spans}
  Let $G$ be a graph and $r>0$ be an integer. Let $W\subseteq V(G)$ be
  such that every vertex in $W$ has degree $r$ in $G$ and $G[W]$ is connected.
  Let $R$ be any $r$-regular graph and $G$ has an induced copy of $R$ on a vertex set $W'$
  containing at least one vertex in $W$. Then $W\subseteq W'$.
\end{observation}
\begin{proof}
  Let $W''$ be $W\setminus W'$. For a contradiction, assume that $W''$ is non-empty.
  It is given that $W\cap W'$ is non-empty, i.e., $W\setminus W''$ is non-empty. 
  Therefore, since $G[W]$ is connected, there exists a vertex $v\in W''$ such that
  $v$ is adjacent to a vertex $u\in W\setminus W''$. Since $u\in W'$ and $G[W']$ 
  induces an $r$-regular graph and $u$
  has degree $r$ in $G$, we obtain that every neighbor of $u$ must be in $W'$. 
  This is a contradiction as $v$ is a neighbor of $u$ and is not in $W'$.
  Hence $W\subseteq W'$.\qed
\end{proof}

\begin{observation}
  \label{obs:threeiso}
  Let $G$ and $G'$ be two graphs such that $|V(G)|=|V(G')|=3$ and $|E(G)|=|E(G')|$.
  Then $G$ and $G'$ are isomorphic.
\end{observation}
\begin{proof}
  If a graph has exactly three vertices, the graph is completely defined
  by its number of edges $e$: If $e=0$, the graph is a null graph,
  if $e=1$, the graph is $K_1\cup K_2$, if $e=2$, the graph is a $P_3$
  and if $e=3$, the graph is a $K_3$.\qed
\end{proof}

\begin{lemma}
  \label{lem:reg-reduction}
  Let $R$ be any connected $r$-regular graph for any $r>2$. 
  Assume that there exists a set of vertices $V'\subseteq V(R)$ such that
  $R[V']$ is a $P_3$ or a $K_3$ and $R - V'$ is connected. 
  Let $R[V']$ be $H'$. Then, there is a linear parameterized reduction from
  \HDED\ to \RED.
\end{lemma}
\begin{proof}
  Let $(G',k)$ be an instance of \HDED. We apply Construction~\ref{con:nonadj} 
  on $(G',k,H=R,V')$ to obtain $G$. 
  We claim that $(G',k)$ is a yes-instance of \HDED\ if and only if $(G,k)$ is a yes-instance of \RED. 

  Let $F'$ be a solution of size at most $k$ of $(G',k)$. We claim that $F'$
  is a solution of $(G,k)$. Let $G''$ be $G - F'$. Assume that the claim is false. 
  Then, there is a set of vertices $U\subseteq V(G'')$ which induces $R$ in $G''$.
  Since $R\setminus V'$ is connected, there is 
  a set of vertices $U'\subseteq U$ which induces $H'$ in $G''$ such that $G''[U\setminus U']$ is a connected graph.  
  Since $G'-F'$ is $H'$-free, at least one vertex $v\in U'$ must be from a branch $V_j$.
  Since $R\setminus V'$ is connected, by the construction, 
  $V_j$ induces a connected graph in $G$ and hence in $G''$. Furthermore,
  every vertex in $V_j$ has degree $r$ in $G''$.
  Now, by Observation~\ref{obs:reg-spans} (invoked with $G=G''$, $W=V_j$ and $W'=U$), every vertex in $V_j$ is in $U$. 
  Since $|V'|=3$, by the construction, $|V_j|=|U|-3$. Hence, by Observation~\ref{obs:reg-dom}
  (invoked with $V'=U\setminus V_j$), $V_j$ is a dominating set in $G''[U]$.
  Therefore, $U=V_j\cup B_j$ where $B_j$ is the set of base vertices of $V_j$ in $G$.
  Since every vertex in $V_j$ has degree $r$ and $G''[U]$ induces an $r$-regular graph, every edge incident to
  the vertices in $V_j$ is in $G''[U]$, i.e., $E_j\subseteq E(G''[U])$, where $E_j$ is the edge set 
  introduced along with $V_j$ in Construction~\ref{con:nonadj}. 
  Now, by an edge counting argument, $E(G''[B_j])$ must have $|E(H')|$ number of edges.
  Therefore, since $|B_j|=3$, by Observation~\ref{obs:threeiso}, $B_j$ induces $H'$ in $G'-F'$, 
  which is a contradiction.
  Lemma~\ref{lem:con:nonadj-backward} proves the converse.\qed
\end{proof}

\begin{observation}
  \label{obs:connected-subgraph}
  Let $G$ be a connected graph with at least $d\geq 1$ vertices. 
  Then, there is a set of vertices $V'\subseteq V(G)$ such that
  $|V'|=d$ and $G[V']$ is connected.
\end{observation}
\begin{proof}
  Let $v$ be any vertex in $G$. Do a breadth first search
  starting from $v$ until $d$ number of vertices are visited.
  Let $V'$ be the set of visited vertices. 
  Clearly, $G[V']$ is connected.\qed
\end{proof}

The following lemma may be of independent interest. The assumption in Lemma~\ref{lem:reg-reduction} comes as a special case of it. 

\begin{lemma}
  \label{lem:peri}
  Let $H$ be any connected graph with minimum degree $d$ for any $d>2$. 
  Then, there exists $V'\subseteq V(H)$ such that $|V'|=d$, $H[V']$ is connected and
  $H\setminus V'$ is connected.
\end{lemma}
\begin{proof}
  Let $\mathcal{H}$ be 
  the set of all connected graphs with $d$ number of vertices. 
  Since the minimum degree of $H$ is $d$, $H$ has at least $d+1$ vertices.
  Hence, by Observation~\ref{obs:connected-subgraph}, 
  there exists at least one $H'\in \mathcal{H}$
  as an induced subgraph of $H$. For a contradiction, assume that for every  $V'\subseteq V(H)$ 
  which induces any $H'\in \mathcal{H}$ in $H$, $H\setminus V'$ is disconnected.
  Among all such sets of vertices, consider a set of vertices 
  $V'\subseteq V(H)$ which induces any $H'\in \mathcal{H}$ in 
  $H$ such that $H - V'$ leaves a component with maximum number of vertices.
  Let the $t>1$ components of $H\setminus V'$ be composed of sets of vertices $V_1, V_2,\ldots, V_t$. 
  Without loss of generality, assume that $H[V_1]$ is a component with maximum
  number of vertices. Every other component has at most $d-1$ vertices. 
  Otherwise, by Observation~\ref{obs:connected-subgraph}, 
  there will be a connected induced subgraph of $d$ vertices 
  in that component deleting which we get a larger component composed of $V_1\cup V'$. 
  Consider $V_j$ for any $j$ such that $2\leq j\leq t$.
  We obtained that $|V_j|\leq d-1$. Hence, the degree of any vertex $v\in V_j$ is at most $d-2$ in $H[V_j]$. 
  Since the minimum degree of $H$ is $d$, there is at least $2$ edges from $v$ to $V'$. 
  Let the neighbourhood of $v$ in $V'$ be $V''$. If none of the vertices in $V''$ is adjacent to $V_1$, 
  then $v$ and any of its $d-1$ 
  neighbours induces a connected graph deleting which gives a larger component. 
  If one of the vertices in $V''$ is adjacent to $V_1$, excluding that we get $d-1$ 
  neighbours of $v$ which along with $v$ induce a connected subgraph and deleting which gives a larger component. 
  This is a contradiction.\qed
\end{proof}

\begin{corollary}
  \label{cor:p3k3}
  Let $H$ be a connected graph with minimum degree 3. 
  Then there exists an induced $P_3$ or $K_3$ with a vertex set $V'$ in $H$ such that $H\setminus V'$ is connected.
\end{corollary}

\begin{theorem}
  \label{thm:regular}
  Let $R$ be a connected regular graph with at least two edges. Then, \RED\ is \NPC. 
  Furthermore, \RED\ is not solvable in time
  $2^{o(k)}\cdot |G|^{O(1)}$, unless ETH fails.
\end{theorem}
\begin{proof}
  Let $R$ be an $r$-regular graph. Since $R$ is connected and has at least 2 edges, $r>1$.
  If $r=2$ then $R$
  is a cycle and the statements follow from Proposition~\ref{pro:bases}(\ref{base:cycle}). 
  Assume that $r\geq 3$. By Corollary~\ref{cor:p3k3}, there exists an induced
  $P_3$ or $K_3$ with a vertex set $V'$
  in $R$ such that $R-V'$ is connected. Now the statements follow from Lemma~\ref{lem:reg-reduction},
  Proposition~\ref{pro:bases}(\ref{base:pted}) and Proposition~\ref{pro:bases}(\ref{base:cycle}).\qed
\end{proof}

The complement graph of a regular graph with at least two non-edges 
is a regular graph with at least two edges. Thus, we obtain the following corollary.

\begin{corollary}
  \label{cor:regular-completion}
  Let $R$ be a regular graph with at least two non-edges. Then, 
  \REC\ is \NPC. Furthermore, \REC\ is not solvable in time
  $2^{o(k)}\cdot |G|^{O(1)}$, unless ETH fails.
\end{corollary}
\section{Handling Disconnected Graphs}
\label{sec:disconnected}

We have seen in Sections~\ref{sec:tree} and \ref{sec:regular} that for any tree or connected regular graph $H$ with 
at least two edges,
\HED\ is \NPC\ and does not admit parameterized subexponential time algorithm
unless ETH fails. In this section, we extend these results to any $H$ with at least two edges such that
$H$ has a largest component which is a tree or a regular graph.

\begin{lemma}
  \label{lem:disconnected}
  Let $H$ be a graph with $t\geq 1$ components. 
  Let $H_1$ be a component of $H$ with maximum number of vertices. 
  Let $H'$ be the disjoint union of 
  all components of $H$ isomorphic to $H_1$. Then, there is a 
  linear parameterized reduction from \HDED\ to \HED.
\end{lemma}
\begin{proof}
  Let $V'\subseteq V(H)$ be the vertex set which induces $H'$ in $H$. 
  Let $(G',k)$ be an instance of \HDED. 
  We apply Construction~\ref{con:nonadj} on $(G',k,H,V')$ to obtain $G$.
  We claim that $(G',k)$ is a yes-instance of \HDED\ if and only if 
  $(G,k)$ is a yes-instance of \HED.

  Let $F'$ be a solution of size at most $k$ of $(G',k)$. 
  For a contradiction, assume that $G-F'$ has an induced $H$ with a vertex set $U$.
  Hence there is a vertex set $U'\subseteq U$ such that $U'$ induces $H'$ in $G - F'$. 
  It is straightforward to verify that a branch vertex can never be part of an induced $H'$ in $G - F'$. 
  Hence $U'$ does not contain a branch vertex and hence $U'$ induces an $H'$ in $G'-F'$, which is a contradiction.
  Lemma~\ref{lem:con:nonadj-backward} proves the converse.\qed
\end{proof}

Lemma~\ref{lem:ident} handles the case of disjoint union of isomorphic connected graphs.

\begin{lemma}
  \label{lem:ident}
  Let $H$ be any connected graph. 
  For every pair of integers $t,s$ such that $t\geq s\geq 1$, there is a linear parameterized reduction from 
  \SHED\ to \THED.
\end{lemma}
\begin{proof}
  The proof is by induction on $t$. 
  The base case when $t=s$ is trivial. 
  Assume that the statement is true for $t-1$, if $t-1\geq s$. 
  Now, we give a linear parameterized reduction from \TOHED\ to \THED. 

  Let $(G',k)$ be an instance of \TOHED. Let $G''$ be a disjoint union of $k+1$
  copies of $H$. Make every pair of vertices $(v_i,v_j)$ adjacent in $G''$ such that 
  $v_i\in V(H_i)$ and $v_j\in V(H_j)$ where $H_i$ and $H_j$ are 
  two different copies of $H$ in $G''$. Let the resultant
  graph be $\hat{G}$. 
  Let $G$ be the disjoint union of $G'$ and $\hat{G}$. We need to prove that
  $(G',k)$ is a yes-instance of \TOHED\ if and only if $(G,k)$ is a yes-instance
  of \THED. 

  Let $F'$ be a solution of size at most $k$ of $(G',k)$. 
  It is straightforward to verify that $\hat{G}$ is $2H$-free.
  Hence, if $G-F'$ has an induced $tH$ then $G'-F'$ has 
  an induced $(t-1)H$, which is a contradiction.
  Conversely, let $(G,k)$ be a yes-instance of \THED.
  Let $F$ be a solution of size at most $k$ of $(G,k)$. 
  For a contradiction, assume that $G'-F$ has an induced $(t-1)H$ with a vertex set $U$.
  Since $|F|\leq k$, $F$ cannot kill all the induced $H$s in $\hat{G}$. Hence, let $U'\subseteq V(\hat{G})$
  induces an $H$ in $G-F$.
  Therefore, $U\cup U'$ induces $tH$ in $G-F$, which is a contradiction.\qed
\end{proof}

Corollary~\ref{cor:ident} is obtained by invoking Lemma~\ref{lem:ident} with $s=1$.
Lemma~\ref{lem:disconnected-final} follows from Lemma~\ref{lem:disconnected} and Corollary~\ref{cor:ident}.
\begin{corollary}
  \label{cor:ident}
  Let $H$ be any connected graph. 
  For every integer $t\geq 1$, there is a linear parameterized reduction from \HED\ to \THED.
\end{corollary}

\begin{lemma}
  \label{lem:disconnected-final}
  Let $H$ be a graph such that $H$ has a component with at least two edges.
  Let $H_1$ be a component of $H$
  with maximum number of vertices. Then there is a linear parameterized reduction
  from \HOED\ to \HED.
\end{lemma}
\begin{proof}
  Let $H'$ be the disjoint union of the components of $H$ which are
  isomorphic to $H_1$. By Lemma~\ref{lem:disconnected}, there is a linear
  parameterized reduction from \HDED\ to \HED.
  Then, by Corollary~\ref{cor:ident}, there is a linear parameterized reduction
  from \HOED\ to \HDED. Composing these two 
  reductions will give a linear parameterized reduction from 
  \HOED\ to \HED.\qed
\end{proof}

\begin{theorem}
  \label{thm:tk2}
  For every $t>1$, \TKTED\ is \NPC. Furthermore, \TKTED\ is not solvable in time
  $2^{o(k)}\cdot |G|^{O(1)}$, unless ETH fails.
\end{theorem}
\begin{proof}
  Follows from Proposition~\ref{pro:bases}(\ref{base:twkted}) and 
  Lemma~\ref{lem:ident} (invoked with $s=2$).\qed
\end{proof}

\begin{theorem}
  \label{thm:final}
  Let $H$ be any graph with at least two edges such that 
  a largest component of $H$ is 
  a tree or a regular graph. Then \HED\ is \NPC. Furthermore, \HED\ is not solvable in time
  $2^{o(k)}\cdot |G|^{O(1)}$, unless ETH fails.
\end{theorem}
\begin{proof}
  Let $H$ be $tK_2\bigcup t'K_1$, for some $t'\geq 0$. Since $H$ has at least two
  edges, $t>1$. Then the statements follow from Theorem~\ref{thm:tk2} and Lemma~\ref{lem:disconnected}.
  If $H$ is not $tK_2\bigcup t'K_1$, let $H_1$ be a largest component which is a tree or a regular graph.
  Clearly, $H_1$ has at least two edges.
  Then, Lemma~\ref{lem:disconnected-final} gives a linear parameterized reduction
  from \HOED\ to \HED. Now, the theorem follows from Theorem~\ref{thm:tree} and Theorem~\ref{thm:regular}.\qed
\end{proof}

Since \HED\ is equivalent to \HBEC, we obtain the following corollary.

\begin{corollary}
  \label{cor:completion}
  Let $\mathcal{H}$ be the set of all graphs $H$ with at least two
  edges such that $H$ has a largest component which is either a tree or a regular graph.
  Let $\overline{\mathcal{H}}$ be the set of graphs such that
  a graph is in $\overline{\mathcal{H}}$ if and only if its complement is 
  in $\mathcal{H}$. Then, for every $H\in \overline{\mathcal{H}}$, \HEC\ is \NPC.
  Furthermore, \HEC\ is not solvable in time
  $2^{o(k)}\cdot |G|^{O(1)}$, unless ETH fails.
\end{corollary}
\section{Concluding Remarks}

We proved that \HED\ is \NPC\ if $H$ is a graph with at least two 
edges and a largest component of $H$ is a tree or a regular graph.
We also proved that, for these graphs $H$, \HED\ cannot be solved in
parameterized subexponential time, unless Exponential Time Hypothesis fails. 
The same results apply for \HBEC.

Assume that we obtain a graph $H'$ from $H$ by deleting 
every vertex with degree $\delta(H)$, the minimum degree of $H$. Also assume that
\HDED\ is \NPC. Then by Lemma~\ref{lem:degree-reduction}, we obtain that
\HED\ is \NPC. The reduction in Lemma~\ref{lem:degree-reduction} 
is not useful if $H'$ is a graph with at most one edge, as for this
\HDED\ is polynomial time solvable. Hence we believe that,
if we can prove the \NP-completeness of \HDED\ where $H'$ is a 
graph in which the set of vertices with degree more than $\delta(G)$
induces a graph with at most one edge, we can prove that \HED\
is \NPC\ if and only if $H$ has at least two edges.

\bibliographystyle{plain}
\bibliography{arxiv}
%

\end{document}